\documentclass[12pt]{article}
\usepackage{amsmath,amsthm,amssymb}

\newcounter{relctr} 
\everydisplay\expandafter{\the\everydisplay\setcounter{relctr}{0}}

\newcommand\labelrel[2]{
	\begingroup
	\refstepcounter{relctr}
	\stackrel{\textnormal{(\alph{relctr})}}{\mathstrut{#1}}
	\originallabel{#2}
	\endgroup
}
\AtBeginDocument{\let\originallabel\label} 
\DeclareMathOperator*{\argmax}{arg\,max}
\DeclareMathOperator*{\argmin}{arg\,min}
\DeclareMathOperator\Fix{Fix}

\def\cP{\mathcal{P}}

\def\R{\mathbb{R}}

\def\op{\mathrm{op}}
\newtheorem{thm}{Theorem}[section]

\newtheorem{cor}[thm]{Corollary}
\newtheorem{lm}[thm]{Lemma}

\theoremstyle{remark}
\newtheorem{rk}[thm]{Remark}

\theoremstyle{definition}
\newtheorem{df}[thm]{Definition}
\newtheorem{eg}[thm]{Example}

\usepackage[backend=bibtex, style=alphabetic, maxbibnames=99]{biblatex}
\DeclareFieldFormat{postnote}{\mknormrange{#1}}
\DeclareFieldFormat{volcitepages}{\mknormrange{#1}}
\DeclareFieldFormat{multipostnote}{\mknormrange{#1}}

\usepackage[hidelinks]{hyperref}
\usepackage{orcidlink}

\providecommand{\keywords}[1]
{
	\small	
	\textbf{\textit{Keywords---}} #1
}
\title{Nash equilibria of games with generalized complementarities}

\date{\today}
\author{Lu Yu\thanks{Université Paris 1 Panthéon-Sorbonne, UMR 8074, Centre d'Economie de la Sorbonne, Paris, France, \href{mailto:yulumaths@gmail.com}{yulumaths@gmail.com}} \,\orcidlink{0000-0001-6154-4229}}

\begin{document}
	\maketitle\begin{abstract}
	To generalize complementarities for games, we introduce some conditions weaker than quasisupermodularity and the single crossing property.  We prove that the Nash equilibria  of a  game satisfying    these conditions form a nonempty complete lattice. This is a purely order-theoretic generalization of Zhou's theorem.
	\end{abstract}	\keywords{Complementarities, Quasisupermodularity, Nash equilibrium, Complete lattice}

\section{Introduction}
Strategic complementarities, a phenomenon ubiquitous in various domains, intricately shape decision-making processes in competitive environments.	For instance, as Samuelson \cite[p.1255]{samuelson1974complementarity} puts it, ``tea and lemon are complements, because tea with lemon makes up our desired brew.'' A joint increase of tea and lemon gives the consumer a benefit exceeding the sum of benefits gained by increasing them separately.

Imprecisely, \emph{strategic complements} defined in \cite[p494]{bulow1985multimarket} refers to the situation where a more aggressive strategy by firm A (e.g.,  lower price, enhanced quality, larger quantity, etc.) raises firm B's marginal profits. In particular, they mutually reinforce one another. To be concrete, consider a simplified mathematical model: a  game with two players (firms A and B) that each has a common twice continuously differentiable payoff function $f(a,b)$ on $\R^2$, where $a$ (resp. $b$) represents the decision of firm A (resp. B). Assume that $f$ is increasing and concave in each argument. We say that  the two decisions are \emph{strategic complements}, if an increase in $a$ raises the marginal payoff $\frac{\partial f}{\partial b}$ of  firm B, i.e.,  if  $\frac{\partial ^2f}{\partial a\partial b}\ge0$. This property is known as  \emph{increasing differences} in \cite[p.42]{topkis1998supermodularity}. A function $g:X\to \R$ on a lattice $X$ is  \emph{supermodular}, if $g(x)+g(y)\le g(x\wedge y)+g(x\vee y)$ for all $x,y\in X$. By Topkis's characterization theorem in \cite[p.1261]{milgrom1990rationalizability}, as $f$ is twice continuously differentiable, the condition $\frac{\partial ^2f}{\partial a\partial b}\ge0$ is also equivalent to the supermodularity of $f:\R^2\to \R$. 

Roughly, a normal form game  is  a \emph{supermodular game} in the sense of \cite[p.179]{topkis1998supermodularity}, if
 the payoff functions are supermodular in own strategy and has increasing differences with others’ strategies. Zhou \cite[Thm.~2]{zhou1994set} (recalled as Corollary \ref{cor:Zhoucomplete}) proves that for a supermodular game where every strategy lattice is compact in a topology finer than the interval topology, the set of Nash equilibria is a nonempty complete lattice.
   Topkis \cite[Thm.~4.2.1]{topkis1998supermodularity} shows the same conclusion for a supermodular game having feasible joint strategies inside Euclidean spaces. Calciano \cite[Theorems 24 and 25]{calciano2010theory} weakens  the increasing difference hypothesis to his \emph{g-modularity} \cite[Def.~11]{calciano2010theory}, but requires every strategy lattice to be a chain. 

The cardinal properties of supermodularity and increasing differences are generalized to ordinal properties, known as \emph{quasisupermodularity} (\cite[p.162]{milgrom1994monotone}) and \emph{single crossing condition} (\cite[p.160]{milgrom1994monotone}) respectively. A normal form game whose payoff functions have the two properties is called a \emph{quasisupermodular game} in \cite[p.179]{topkis1998supermodularity}. The class of quasisupermodular games encompasses a broader range of scenarios compared to supermodular games, while retaining their fundamental characteristics, i.e., the existence and order structure of Nash equilibria.

Inspired by the work \cite{licalzi1992subextremal,veinott1992lattice,agliardi2000generalization}, we propose several conditions strictly weaker than quasisupermodularity and single crossing property used in \cite{milgrom1994monotone}. These extensions require less symmetries, so they  define  classes of games with general complementarities, strictly wider than the class of quasisupermodular games. In the real world, two players with complementarities may be asymmetric. We prove in Theorem \ref{thm:existence} that such games admit Nash equilibria. With some extra conditions,  we show that the set of Nash equilibria  admits a largest element (Theorem \ref{thm:largest}), or forms a nonempty complete lattice (Theorem \ref{thm:structure}).

Section \ref{sec:varquas} recalls the notions of quasisupermodularity and its variants in the literature. We introduce several extensions used in Section \ref{sec:mainTopkis3} and compare them to the classical notions. In Section \ref{sec:whymaxexists}, we show that one extension guarantees the existence of maximum of a function. In the setting of game theory, it becomes the existence of best response of a player.  In Section \ref{sec:mainTopkis3}, we  introduce variants of the single crossing condition. These notions and the   existence of maximum from Section \ref{sec:whymaxexists} are applied to study Nash equilibria of games with general complementarities.  
\subsection*{Notation and conventions}
A set with a partial order is called a \emph{poset}. For a poset $(X,\le)$, let $(X^{\op},\lesssim)$ denote the opposite poset, i.e., it has the same underlying set as $X$, with $y\lesssim x$ in $X^{\op}$ iff $x\le y$ in $X$. For a map $f:X\to Y$ between posets and a value $y\in Y$, let $[f\le y]:=\{x\in X|f(x)\le y\}$ and $[f\ge y]:=\{x\in X|f(x)\ge y\}$ be the lower and the upper level sets respectively.
\section{Variants of quasisupermodularity}\label{sec:varquas}
Quasisupermodularity introduced in \cite[p.162]{milgrom1994monotone} is used to define a class of games with strategic complements, called quasisupermodular games. We review  several  variants of quasisupermodularity in the literature, such as subextremality of LiCalzi and Veinott, and compare them. We also introduce meet-subextremality, then explain in Lemma \ref{lm:sub=w+v} why it is ``half  subextremality".

Let $C$ be a chain, $X$ be a lattice, and $f:X\to C$ be a map.
\begin{df}\label{df:supextremal}
	If for any $x,y\in X$, 
	\begin{enumerate}
		\item\label{it:qsup} (\cite[p.162]{milgrom1994monotone}) the condition $f(x)\ge f(x\wedge y)$ (resp. $f(x)\le f(x\wedge y)$) implies $f(x\vee y)\ge f(y)$ (resp. $f(x\vee y)\le f(y)$), and if the condition $f(x)>f(x\wedge y)$ (resp. $f(x)<f(x\wedge y)$) implies $f(x\vee y)>f(y)$ (resp. $f(x\vee y)<f(y)$), then $f$ is called quasisupermodular (resp. quasisubmodular);
		
		\item the condition $f(x\wedge y)<f(x)$ implies  $f(x)\wedge f(y)<f(x\vee y)$, and if the condition $f(x\vee y)<f(x)$ implies $f(x)\wedge f(y)<f(x\wedge y)$, then $f$ is called weakly quasisupermodular;
		\item (\cite[Def. p.252]{agliardi2000generalization}) the condition $f(x)\vee f(y)\ge f(x\wedge y)$ implies $f(x\vee y)\ge f(x)\wedge f(y)$, and if the condition $f(x)\vee f(y)> f(x\wedge y)$ implies $f(x\vee y)> f(x)\wedge f(y)$, then $f$ is called pseudo-supermodular;
		\item the condition $f(x\vee y)<f(x)\wedge f(y)$ implies $f(x\wedge y)>f(x)$, then $f$ is called weakly pseudo-supermodular;

		\item \textup{(\cite[p.4]{licalzi1992subextremal})} either \begin{equation}\label{eq:subext1}
			f(x\wedge y)\vee f(x\vee y)\le f(x)\vee f(y)\end{equation} or \begin{equation}\label{eq:subext2}f(x\wedge y)\wedge f(x\vee y)\le f(x)\wedge f(y),
		\end{equation} (resp. either \[
		f(x\wedge y)\vee f(x\vee y)\ge f(x)\vee f(y)\] or \[f(x\wedge y)\wedge f(x\vee y)\ge f(x)\wedge f(y)\text{,)}
		\] then $f$ is called subextremal (resp. superextremal);
		\item\label{it:latticesubextremal} (\cite[Ch.~6, Sec.~5]{veinott1992lattice}, \cite[p.5]{licalzi1992subextremal}) there is $t<f(y)$ (resp. $t>f(y)$) in $C$ such that either \[f(x\wedge y)\vee f(x\vee y)\le f(x)\vee t\] or \[f(x\wedge y)\wedge f(x\vee y)\le f(x)\wedge t\] (resp. either \[f(x \wedge y)\vee f(x\vee y)\ge f(x)\vee t\] or \[f(x\wedge y)\wedge f(x\vee y)\ge f(x)\wedge t,)\] then $f$ is called lattice subextremal (resp. lattice superextremal).
		\item either $f(x\wedge y)\le f(x)$ or $f(x\vee y)\le f(x)\vee f(y)$ (resp. either $f(x\wedge y)\ge f(x)$ or $f(x\vee y)\ge f(x)\wedge f(y)$), then $f$ is called meet-subextremal (resp. meet-superextremal);
		\item either $f(x\wedge y)\le f(x)\vee f(y)$ or $f(x\vee y)\le f(x)$ (resp. either $f(x\wedge y)\ge f(x)\wedge
		f(y)$ or $f(x\vee y)\ge f(x)$), then $f$ is called  join-subextremal (resp. join-superextremal);
	\end{enumerate}
\end{df}
If $f$ is quasisupermodular, then $f$ is pseudo-supermodular (\cite[Prop.~2]{agliardi2000generalization}) and weakly quasisupermodular.   If $f$ is pseudo-supermodular, then $f$ is weakly pseudo-supermodular. If $f$ is weakly pseudo-supermodular, then $f$ is meet-superextremal.
\begin{rk}In Definition \ref{df:supextremal}, 
	if $X$ is a chain, then  all conditions are satisfied except possibly \ref{it:latticesubextremal}.
\end{rk}
\begin{rk} If $x,y$ in Definition \ref{df:supextremal} \ref{it:latticesubextremal} are comparable and $C$ has no minimum, then the condition is always satisfied. (By symmetry, we may assume  $x\le y$. If $f(x)<f(y)$, then take $t=f(x)$ and $f(x\wedge y)\wedge f(x\vee y)\le f(x)\wedge t$. If $f(x)\ge f(y)$, then take any $t\in C$ with $t<f(y)$.  One has $f(x\wedge y)\vee f(x\vee y)\le f(x)\vee t$.)\end{rk}

\begin{lm}\label{lm:wksub}
	The following conditions are equivalent:
	\begin{enumerate}
		\item the map $f$ is  meet-subextremal;
		\item the  map $X^{\op}\to C$ induced by $f$ is  join-subextremal;
		\item the  map $X\to C^{\op}$ induced by $f$ is  meet-superextremal;
		\item the  map $X^{\op}\to C^{\op}$ induced by $f$ is  join-superextremal. 
	\end{enumerate}
\end{lm}
\begin{proof}
	By definition.
\end{proof}

\begin{lm}\label{lm:lat>qua}
	If  $f$ is lattice superextremal, then $f$ is quasisupermodular.
\end{lm}
\begin{proof}Assume that $f$ is not quasisupermodular. Then there exist $x,y\in X$, such that at least one of the following cases holds: \begin{gather}
		f(x)\ge f(x\wedge y),\quad f(x\vee y)<f(y);\label{eq:A}\\
		f(x)> f(x\wedge y),\quad f(x\vee y)\le f(y).\label{eq:B} 
	\end{gather}
	Since $f$ is lattice superextremal, there is $t>f(y)$ such that at least one of the following cases holds:
	\begin{gather}
		f(x \wedge y)\vee f(x\vee y)\ge f(x)\vee t\text{ or }\label{eq:1}\\
		f(x\wedge y)\wedge f(x\vee y)\ge f(x)\wedge t.\label{eq:2}
	\end{gather}
	\begin{enumerate}
		\item Assume (\ref{eq:A}) and (\ref{eq:1}). We have $f(x \wedge y)\vee f(x\vee y)\ge t>f(y)>f(x\vee y)$, so $f(x \wedge y)>f(x\vee y)$. Thus, $f(x)\ge f(x \wedge y)\ge f(x)\vee t\ge f(x)$. It implies $f(x)= f(x \wedge y)\ge t>f(y)>f(x\vee y)$. For every $t'>f(x)$ in $C$,  one has $f(y)\wedge t'=f(y)>f(x\vee y)\ge f(x\wedge y)\wedge f(x\vee y)$ and $f(y)\vee t'=t'>f(x)=f(x \wedge y)\vee f(x\vee y)$. This contradicts that $f$ is  lattice superextremal.
		\item  Assume (\ref{eq:A}) and (\ref{eq:2}). We have $t>f(y)>f(x\vee y)\ge f(x\wedge y)\wedge f(x\vee y)\ge f(x)\wedge t$. Thus, $t>f(x)$ and $f(x\wedge y)\ge f(x\wedge y)\wedge f(x\vee y)\ge f(x)\ge f(x\wedge y)$. Hence $f(x)=f(x\wedge y)\le f(x\vee y)<f(y)$. For every $t'>f(x)$ in $C$,  one has $f(y)\vee t'\ge f(y)>f(x\vee y)=f(x \wedge y)\vee f(x\vee y)$ and $f(y)\wedge t'>f(x)=f(x\wedge y)\wedge f(x\vee y)$. This contradicts that $f$ is  lattice superextremal.
		\item Assume (\ref{eq:B}) and (\ref{eq:1}). Then $f(x \wedge y)\vee f(x\vee y)\ge f(x)\vee t\ge f(x)>f(x\wedge y)$. Thus, $f(x\vee y)>f(x \wedge y)$ and $f(x)\vee t\ge t>f(y)\ge f(x\vee y)\ge f(x)> f(x\wedge y)$. This contradicts (\ref{eq:1}).
		\item Assume (\ref{eq:B}) and (\ref{eq:2}). Then $f(x)>f(x\wedge y)\ge 	f(x\wedge y)\wedge f(x\vee y)\ge f(x)\wedge t$ and $t>f(y)\ge f(x\vee y)\ge 	f(x\wedge y)\wedge f(x\vee y)\ge f(x)\wedge t$. Thus, $f(x)\wedge t> f(x)\wedge t$. A contradiction.
	\end{enumerate} 
\end{proof}
Lemma \ref{lm:qua>lat} shows that when $C=\R$,  quasisupermodularity coincides with lattice superextremality.
\begin{lm}\label{lm:qua>lat}
	Suppose that $C$ has no maximum, then $f$ is quasisupermodular if and only if $f$ is lattice superextremal.
\end{lm}
\begin{proof}By Lemma \ref{lm:lat>qua}, it remains to prove that  quasisupermodularity implies  lattice superextremality. 
	
	Assume the contrary that $f$ is not lattice superextremal. Then there exist $x,y\in X$ such that for every $t>f(y)$, one has $f(x \wedge y)\vee f(x\vee y)< f(x)\vee t$ and $f(x\wedge y)\wedge f(x\vee y)< f(x)\wedge t$. Since $f(y)$ is not the maximum of $C$, there is $t_1\in C$ with $t_1>f(y)$.   There are exactly two cases.
	\begin{enumerate}\item $f(x\wedge y)\ge f(x\vee y)$: Then $f(x\vee y)< f(x)\wedge t_1\le f(x)$. As $f$ is quasisupermodular, one has $f(x\wedge y)>f(y)$. Take $t_2=f(x\wedge y)$. Then $f(x\wedge y)<f(x)\vee t_2=f(x)$. As $f$ is quasisupermodular, one has $f(y)<f(x\vee y)$. Take $t_3=f(x\vee y)$. Then   $f(x\vee y)< f(x)\wedge t_3\le t_3=f(x\vee y)$, which is a contradiction.
		
		\item $f(x\wedge y)< f(x\vee y)$: Then $f(x\wedge y)< f(x)\wedge t\le f(x)$. As $f$ is quasisupermodular, one has $f(x\vee y)>f(y)$. Take $t_4=f(x\vee y)$. Then $f(x\vee y)<f(x)\vee t_4=f(x)$. As $f$ is quasisupermodular, one has $f(y)<f(x\wedge y)$. Take $t_5=f(x\wedge y)$. Then $f(x\wedge y)< f(x)\wedge t_5\le t_5=f(x\wedge y)$, which is a contradiction.\end{enumerate}
\end{proof}
\begin{rk} It is  stated in \cite[Ch.~6, Lem.~34]{veinott1992lattice} that  if $f$ is quasisubmodular, then it is lattice subextremal. This statement seems questionable if $C$ has a minimum element $u$. For example, the constant map $f:X\to C, \quad x\mapsto u$ is quasisubmodular but not lattice subextremal.
\end{rk}
\begin{lm}\label{lm:sub=w+v}
	The following conditions are equivalent:
	\begin{enumerate}
		\item \label{it:sub}$f$ is subextremal;
		\item  \label{it:wkalm}$f$ is  meet-subextremal and  join-subextremal.
	\end{enumerate}
\end{lm}
\begin{proof}
	\begin{itemize}
		\item Assume Condition \ref{it:sub}. For any $x,y\in X$ with $f(x\vee y)> f(x)\vee f(y)$,  (\ref{eq:subext1}) does not hold. Consequently, (\ref{eq:subext2}) is true. Since $f(x\vee y)> f(x)\wedge f(y)$, we have $f(x\wedge y)\le f(x)\wedge f(y)\le f(x)$. This proves that $f$ is  meet-subextremal.
		
		The induced map $X^{\op}\to C$ is subextremal, hence  meet-subextremal. By Lemma \ref{lm:wksub}, the map $f:X\to C$ is  join-subextremal. 	
		
		\item Suppose that $f$ verifies Condition \ref{it:wkalm} but not Condition \ref{it:sub}. Then there exist $x,y\in X$ with $f(x\vee y)\vee f(x\wedge y)>f(x)\vee f(y)$ and $f(x\vee y)\wedge f(x\wedge y)>f(x)\wedge f(y)$. By symmetry, we may assume $f(x)\le f(y)$. Then $f(x)<f(x\vee y)$.  We have $f(x)<f(x\wedge y)\labelrel\le{myeq:usejoin}f(y)<f(x\vee y)\vee f(x\wedge y)=f(x\vee y)$, where \eqref{myeq:usejoin} uses join-subextremality. However, this contradicts meet-subextremality.
		\end{itemize}
	\end{proof}
Corollary \ref{cor:qs>supextremal} shows that superextremal maps are generalizations of quasisupermodular maps.
\begin{cor}\label{cor:qs>supextremal}
	If $f$ is weakly quasisupermodular, then it is superextremal.
\end{cor}
\begin{proof}
	For any $x,y\in X$, if $f(x\wedge y)<f(x)$, then $f(x\vee y)>f(x)\wedge f(y)$, so $f$ is meet-superextremal. Similarly, if $f(x\vee y)<f(x)$, then $f(x)\wedge f(y)<f(x\wedge y)$, so $f$ is join-superextremal. By dual of Lemma \ref{lm:sub=w+v}, $f$ is superextremal.
\end{proof}The converse of Corollary \ref{cor:qs>supextremal} is not true, as Example \ref{eg:01ab} \ref{it:superextremalnotwkqs} shows. We provide a series of examples to display the relation between different concepts of quasisupermodularity in the literature.
\begin{eg}\label{eg:01ab}Let $X=\{0,1,a,b\}$, where $a,b$ are incomparable, $a\wedge b=0$ and $a\vee b=1$. Then $X$ is a lattice.	
	\begin{enumerate}
		\item \label{it:wsubnotv} Define $f:X\to \R$ by $f(0)=0$, $f(a)=2$, $f(b)=f(1)=1$. Then $f$ is  meet-superextremal. However, $f(a\wedge b)>f(a)\vee f(b)$ and $f(a\vee b)>f(a)$, so $f$ is not join-superextremal.
		By Corollary \ref{cor:qs>supextremal}, it is not quasisupermodular.
		\item\label{it:superextremalnotwkqs} Define $g:X\to \R$ by $g(1)=g(a)=g(b)=1$, $g(0)=0$. Then $g$ is superextremal. As $g(a\wedge b)< g(a)$ and $g(a)\wedge g(b)\ge g(a\vee b)$, the function $g$ is not weakly quasisupermodular.
		\item Define $h:X\to \R$ by $h(0)=h(a)=2$, $h(b)=1$, $h(1)=0$. Then $h$ is weakly quasisupermodular but not weakly pseudo-supermodular. Indeed, $h(1)<h(a)\wedge h(b)$, but $h(0)=h(a)$.
		\item Define $u:X\to \R$ by $u(0)=0$ and $u(a)=u(b)=u(1)=1$. Then $u$ is weakly pseudo-supermodular but neither pseudo-supermodular nor weakly quasisupermodular. Indeed, one has $u(0)<u(a)$ but $u(1)=u(a)\wedge u(b)$. One has $u(a)\vee u(b)>u(0)$ but $u(1)=u(a)\wedge u(b)$.
		\item\label{it:wkqsnotqs} Define $v:X\to \R$ by $v(0)=v(a)=2$, $v(b)=1$ and $v(1)=0$. Then $v$ is weakly quasisupermodular, but not weakly pseudo-supermodular. Indeed, one has $v(a)\ge v(a\wedge b)$ and $v(a\vee b)<v(a)\wedge v(b)$.
		\item Define $w:X\to \R$ by $w(0)=0$, $w(b)=1$, $w(1)=2$ and $w(a)=3$. Then $w$ is pseudo-supermodular but neither weakly quasisupermodular nor join-superextremal.
\end{enumerate}\end{eg}

\section{Existence of maximum}\label{sec:whymaxexists}
In a normal form game, each player's optimal response to the strategies chosen by the other players is to maximize their respective payoff functions. Milgrom and Shannon \cite[Thm.~A4]{milgrom1994monotone} prove that the maximizers of a quasisupermodular, order-theoretically upper semicontinuous function on a complete lattice attains is a nonempty lattice. It  is an analog of the  topological fact (see, e,g., \cite[p.991]{ceder1963compactness}) that  an upper semicontinuous function on a compact space attains its maximum. In Theorem \ref{thm:41}, we  give a purely order-theoretical sufficient condition for the existence of minimum. Its dual statement about maximum is similar to the Milgrom-Shannon theorem. Example \ref{eg:01ab} \ref{it:superextremalnotwkqs} satisfies the hypotheses of the dual statement, but the set of maximizers $\{a,b,1\}$ is not a lattice. Still, it is a quasisublattice.
\begin{df}
\begin{itemize}
\item \cite[p.12]{licalzi1992subextremal} Let $X$ be a lattice. A subset $S$  of $X$ is  a \emph{quasisublattice}, if for any $x,y\in S$, at least one of $x\wedge y$ and  $x\vee y$ is in $S$.
\item \cite[p.542]{kukushkin2013increasing} A poset $P$ is \emph{chain-complete upwards} (resp. \emph{downwards}) if $\sup_PC$ (resp. $\inf_PC$) exists for every nonempty chain $C\subset P$.
\end{itemize}	
\end{df}

In Theorem \ref{thm:41}, we generalize \cite[Ch.~6, Theorem 41]{veinott1992lattice}, which assumes furthermore that the chain $C$ is complete and for every $t\in C$, the set $[f\le t]$ is \emph{chain-subcomplete} in the lattice $X$. 
This generalization is proper, as shown by Example \ref{eg:41}.
\begin{thm}\label{thm:41}
	Let $f:X\to C$ be a  meet-subextremal  map from  a nonempty lattice to  a chain. Suppose that for every $t\in C$, the set $[f\le t]$  is chain-complete downwards  and admits maximal elements. Then $\argmin_X f$  is a nonempty  quasisublattice of $X$.
\end{thm}
A correspondence $F:X\to 2^Y$ from a poset $X$ to a lattice $Y$ is called \emph{weakly ascending}, if for any $x<x'$ in $X$, every $y\in F(x)$ and every $y'\in F(x')$, one has either $y\vee y'\in F(x')$ or $y\wedge y'\in F(x)$.
\begin{proof}Define a correspondence $F:f(X)\to 2^X, \quad F(t)=[f\le t]$. By assumption, for every $t\in f(X)$, the value $F(t)$ is chain-complete downwards  and admits maximal elements.
	
	We prove that $F$ is weakly ascending. For this, consider any $t<\tau$ in $f(X)$, every $s\in F(t)$ and every $\sigma\in F(\tau)$. As $f$ is meet-subextremal, one has either $f(s\wedge \sigma)\le f(s)\le t$ or $f(s\vee \sigma)\le f(\sigma)\vee f(s)\le \tau$. Then either $s\wedge \sigma\in F(t)$ or $s\vee \sigma\in F(\tau)$. 
	
	Let $\cP$ be the set of increasing selections of $F$. By a dual of \cite[Theorem 2.2]{kukushkin2013increasing}, as $F$ is weakly ascending, $\cP$ is nonempty. For $r,r'\in\cP$, by  $r\le r'$ we mean  $r(t)\le r'(t)$ for all $t\in f(X)$. Then $\cP$ is a poset under the relation $\le$.
	
	We show that every nonempty chain  in $\cP$ has a lower bound.  Let  $\{r^i\}_{i\in I}$ be such a chain.  For every $t\in f(X)$, the subset $\{r^i(t)\}_{i\in I}$ of $F(t)$ is a chain. Because $F(t)$ is chain-complete downwards, $r(t):=\inf_{F(t)}\{r^i(t)\}_{i\in I}$ exists. For any $t\le t'$ in $f(X)$ and every $i\in I$, one has $r(t)\le r^i(t)\le r^i(t')$. From $r(t)\in F(t)\subset F(t')$, one has $r(t)\le \inf_{F(t')}\{r^i(t')\}_{i\in I}=r(t')$. Therefore, $r\in \cP$ is a lower bound on the chain $\{r^i\}_i$. 
	
	By Zorn's lemma, $\cP$  has a minimal element $r$. For any $t\le \tau$ in $f(X)$,  one has $r(t)\in F(t)\subset F(\tau)$. Then $\{r(t)\}_{t\le \tau,t\in f(X)}$ is a nonempty  chain in $F(\tau)$. Since $F(\tau)$ is chain-complete downwards, the element $m(\tau):=\inf_{F(\tau)}\{r(t)\}_{t\le \tau,t\in f(X)}$ exists. 
	
	For any $\tau\le \tau'$ in $f(X)$, by $F(\tau)\subset F(\tau')$, one has \[m(\tau)=\inf_{F(\tau)}\{r(t)\}_{t\le \tau,t\in f(X)}\le \inf_{F(\tau')}\{r(t)\}_{t\le \tau,t\in f(X)}=\inf_{F(\tau')}\{r(t)\}_{t\le \tau',t\in f(X)}=m(\tau').\] Thus,  the map $m:f(X)\to X$ is increasing and $m\in\cP$.
	
	For every $t\in f(X)$, one has $m(t)\le r(t)$. Consequently, $m\le r$ in $\cP$. Since $r$ is minimal in $\cP$, one has $m=r$. Then for any $t'\le t$ in $f(X)$, one has $r(t)=m(t)\le r(t')\le r(t)$. As a consequence, one has $r(t)=r(t')$. Because $C$ is a chain, there is $s_0\in X$ such that the map $r:f(X)\to X$ is constantly $s_0$. Then $s_0\in F(t)$, so $f(s_0)\le t$. Therefore, $f(s_0)=\min f$. In particular, $\argmin_Xf$ is nonempty.
	
	For any $x,y\in \argmin_Xf$, as $f$ is  meet-subextremal, either $f(x\wedge y)\le f(x)=\min f$ or $f(x\vee y)\le f(x)\vee f(y)=\min f$, i.e., either $x\wedge y$ or $x\vee y$ is in $\argmin_Xf$. This shows that $\argmin_Xf$ is a quasisublattice of $X$.
\end{proof}
\begin{eg}\label{eg:41}
	Let $f:[0,2]\to \R$ be the characteristic function of $[1,2)$. Then $f$ is subextremal. By Lemma \ref{lm:sub=w+v}, it satisfies the condition of Theorem \ref{thm:41}. However, $[f\le0]=[0,1)\cup\{2\}$ is not chain-subcomplete upwards in $[0,2]$. In this case, \cite[Ch.~6, Theorem 41]{veinott1992lattice} is not applicable.
\end{eg}

\begin{cor}\label{cor:41}
	Let $f:X\to C$ be a  map from  a nonempty lattice to  a chain.
	\begin{enumerate}
		\item If $f$ is   join-subextremal and for every $t\in C$, $[f\le t]$ is chain-complete upwards in $X$ and admits minimal elements, then $\argmin_X f$ is a nonempty  quasisublattice of $X$.
		\item \label{it:wsup}	If $f$ is  meet-superextremal and for every $t\in C$, $[f\ge t]$ is chain-complete downwards  in $X$ and admits maximal elements, then $\argmax_X f$ is a nonempty  quasisublattice of $X$.
		\item\label{it:vsup} Assume that $f$ is  join-superextremal and  for every $t\in C$, the subset $[f\ge t]\subset X$ is chain-complete upwards  and admits minimal elements. Then $\argmax_X f$ is a nonempty  quasisublattice of $X$.
	\end{enumerate}
\end{cor}
\begin{proof}
	They follow from Theorem \ref{thm:41} and Lemma \ref{lm:wksub}. 
\end{proof}
In Corollary \ref{cor:41} \ref{it:wsup}, if furthermore $C=\R$, and if the function $f$ is  order upper semi-continuous (in the sense of \cite[p.1261]{milgrom1990rationalizability}) and quasisupermodular, then the result  specializes to the existence part of \cite[Theorem A4]{milgrom1994monotone}. This generalization is also proper, as shown by Example \ref{eg:01ab} \ref{it:wsubnotv}.
\section{Structure of Nash equilibria}\label{sec:mainTopkis3}
We  study the set of Nash equilibria of normal form games with generalized complementarities. A brief overview of the main results is as follows. In Theorem \ref{thm:existence}, we establish the existence of Nash equilibria. With stronger hypotheses, we enhance it to the existence of largest Nash equilibrium in Theorem \ref{thm:largest}. We give sufficient conditions ensuring that Nash equilibria form a complete lattice in Theorem \ref{thm:structure}.
\subsection*{Model}
Recall the notion of Nash equilibria of normal form games.
\begin{df}
	A normal form game $(N,\{S_i\},\{C_i\},\{u_i\})$ is the following data: \begin{enumerate}
		\item a nonempty  set of players $N$;
		\item for every $i\in N$, a nonempty set $S_i$ of the strategies  of player $i$; Write $S=\prod_iS_i$ for the set of joint strategies and $S_{-i}:=\prod_{j\neq i}S_j$;
		\item for every player $i\in N$,  a nonempty chain $C_i$ of gains of player $i$ and a payoff function $u_i: S\to C_i$.
	\end{enumerate}
\end{df}
Fix such a game $(N,\{S_i\},\{C_i\},\{u_i\})$.	
\begin{df}A  joint strategy $x\in  S$ is a Nash equilibrium  if $u_i(y_i,x_{-i})\le u_i(x)$ for every $i\in N$ and every $y_i\in S_i$. \end{df}
for every player $i\in N$, the (individual) best response correspondence  $R_i:S\to 2^{S_i}$ is defined by \[R_i(x)=\argmax_{y_i\in S_i}u_i(y_i,x_{-i}).\] Since $R_i$ factors through the natural projection $S\to S_{-i}$, it can also be written as  $R_i:S_{-i}\to 2^{S_i}$. The joint best response $R:S\to 2^S$ is defined as $R(x)=\prod_{i\in N}R_i(x_{-i})$. The set of Nash equilibria of the game coincides with the set $\Fix(R)$ of fixed points of $R:S\to 2^S$.

The intuitive idea of two parts having  complementarities is  that,  increasing the level of one part makes desire to increase the  level of the other as well. Single crossing property is a notion to capture the idea that ``marginal returns to increasing one’s strategy rise with increases in competitors’ strategies'' (\cite[Footnote 6]{dubey2006strategic}). Calciano \cite[Def.~11]{calciano2010theory} proposes a notion called  \emph{generalized modularity} to capture the idea. Inspired by his work, we propose two related conditions and recall their precedent from \cite{milgrom1994monotone}.
\begin{df}
	Let $X,T$ be  posets, $C$  be a chain. Consider a map $f:X\times T\to C$.
	If for any $x<x'$ in $X$ and any $t< t'$ in $T$,
	\begin{enumerate}
		\item the condition $f(x,t)\le f(x',t)$  implies $f(x,t')\le f(x',t')$, then $f$ is called  modular-crossing  relative to $(X,T)$; 
		\item	the condition $f(x,t)\le f(x',t)$ (resp. $f(x,t')\le f(x',t')$) implies the existence of $u\ge x'$ in $X$ (resp. $v\le x$ in $X$) with
		$f(x,t')\le f(u,t')$ (resp. $f(x,t)\le f(v,t)$), then $f$ is called upper-crossing (resp. lower-crossing)  relative to $(X,T)$;
		\item \cite[p.160]{milgrom1994monotone} the condition $f(x,t)\le f(x',t)$ implies  $f(x,t')\le f(x',t')$, and if the condition $f(x,t)< f(x',t)$  implies $f(x,t')< f(x',t')$, then $f$ is said to satisfy the single crossing property relative to $(X,T)$.
	\end{enumerate}
\end{df}
If $f$ satisfies the single crossing property, then $f$ is modular-crossing. If $f$ is modular-crossing, then $f$ is  upper-crossing.
\begin{eg}\begin{enumerate}
\item	Let $X=T=\{0,1\}$ and $C=\R$. Define a function $f:X\times T\to C$  by $f(0,0)=f(0,1)=f(1,1)=0$ and $f(1,0)=1$. Then $f$ is modular-crossing but does not satisfy the single crossing property relative to $(X,T)$.

\item	Let $X=\{0,1,2\}$, $T=\{0,1\}$ and $C=\R$. Define a function $f:X\times T\to C$ by setting $f(1,1)=-1$ and  all the other values to be $0$. Then $f$ is upper-crossing, but not modular crossing relative to $(X,T)$.\end{enumerate}
\end{eg}
\subsection*{Existence}
In	Theorem \ref{thm:existence}, we consider the existence of Nash equilibria under purely order-theoretic hypotheses. The meet-superextremality (resp. modular-crossing) assumption is strictly weaker than quasisupermodularity (resp. single crossing property) used in \cite[Thm.~12]{milgrom1994monotone}. 
\begin{thm}\label{thm:existence}
	Assume  that for every $i\in N$, \begin{enumerate}\item $S_i$ is a complete lattice;
		\item\label{it:uisuperextremal}   for every $s_{-i}\in S_{-i}$, the function $u_i(\cdot,s_{-i}):S_i\to C_i$ is  meet-superextremal (resp. join-superextremal);
			\item\label{it:modularcrossing}   the payoff function $u_i:S\to C_i$ is  modular-crossing  relative to $(S_i,S_{-i})$;
		\item\label{it:uiuplevel} for  every $s_{-i}\in S_{-i}$ and every $t\in C_i$, the subset $[u_i(\cdot,s_{-i})\ge t]$ is chain-complete downwards (resp. upwards) and admits maximal (resp. minimal) elements.
	\end{enumerate}Then the game admits a Nash equilibrium.
\end{thm}

\begin{proof}
	By symmetry, it is enough to prove the statement  without parentheses.	From Theorem \ref{cor:41} \ref{it:wsup}, Conditions \ref{it:uisuperextremal} and \ref{it:uiuplevel}, for every $i\in N$ and every $s\in S$, the subset $R_i(s)\subset S_i$ is \emph{nonempty} chain-complete downwards, and it admits maximal elements.
	
	We show that $R_i:S_{-i}\to 2^{S_i}$ is weakly ascending. Assume the contrary. Then there exist $x_{-i}<x'_{-i}$ in $S_{-i}$,  $x_i\in R_i(x_{-i})$ and $x'_{-i}\in R_i(x'_{-i})$ satisfying $x_i\wedge x'_i\notin R_i(x_{-i})$ and $x_i\vee x'_i\notin R_i(x'_{-i})$.
	Then  $u_i(x_i\wedge x'_i,x_{-i})<u_i(x_i,x_{-i})$. Because $x_i\in R_i(x_{-i})$, one has   $u_i(x_i,x_{-i})\ge u_i(x'_i,x_{-i})$. By Condition \ref{it:uisuperextremal}, $u_i(x_i\vee x_i',x_{-i})\ge u_i(x_i,x_{-i})\wedge u_i(x'_i,x_{-i})=u_i(x'_i,x_{-i})$. 
	Since $x_i\wedge x'_i\notin R_i(x_{-i})$, we get $x_i\wedge x'_i\neq x_i$. Thus, $x_i'<x_i\vee x_i'$. From Condition \ref{it:modularcrossing}, one has $u_i(x_i\vee x_i',x'_{-i})\ge u_i(x'_i,x'_{-i})$. This contradicts $x_i\vee x'_i\notin R_i(x'_{-i})$.

	By \cite[Theorem 2.2]{kukushkin2013increasing}, as $R_i:S_{-i}\to 2^{S_i}$ is weakly ascending, there is an increasing selection $r_i:S_{-i}\to S_i$ of $R_i$.
	For every $s\in S$, let $r(s)=(r_i(s_{-i}))_{i\in N}$ be an element of $S$ whose $i$-th coordinate is $r_i(s_{-i})$. Thus,  $r:S\to S$ is
	a selection of the best reply correspondence $R:S\to 2^S$. By construction, $r$ is increasing. By completeness of $S$ and Tarski's fixed point theorem \cite[Theorem 1]{tarski1955lattice}, $r$ has a fixed point, which is a Nash equilibrium.
\end{proof}
In Theorem \ref{thm:existence}, the set of Nash equilibria may not have largest nor least element.
\begin{eg}
	Let $N=\{1,2\}$ and $S_i=[0,1]$. Define $u_1:S\to \R$ by $u_1(s_1,s_2)=s_1s_2$
 Define \[u_2:S\to \R,\quad s\mapsto \begin{cases}
0 & \text{ if }s_2\in\{0,1\},\\
1 & \text{ if }0<s_2<1.
 \end{cases}\] Then the corresponding game satisfies the condition of Theorem \ref{thm:existence}. The best reply correspondence is \[R:S\to 2^S,\quad s\mapsto \begin{cases}
	[0,1]\times (0,1) & \text{ if } s_2=0,\\
	\{1\}\times (0,1) & \text{ if }s_2>0.
	\end{cases}\] The set of Nash equilibria is $\{1\}\times (0,1)$, which has no largest nor least element.
\end{eg}
In Theorem \ref{thm:existence}, the set of Nash equilibria may not be a lattice.
\begin{eg}
	Let $N=\{1,2\}$ and $S_i=[0,1]$. Define $u_1:S\to \R$ by \[u_1(s)=\begin{cases}
1& \text{ if }s_2\le 1/2\le s_1,\\
0 &\text{else}.
	\end{cases}\]
Let $u_2:S\to \R$ be constantly zero. Then the corresponding game satisfies the condition of Theorem \ref{thm:existence}. The best reply correspondence is \[R:S\to 2^S,\quad s\mapsto \begin{cases}
		[1/2,1]\times [0,1] & \text{ if }s_2\le 1/2,\\
		 [0,1]^2& \text{ if }s_2>1/2.
	\end{cases}\] The set of Nash equilibria is $[1/2,1]\times [0,1/2]\cup[0,1]\times (1/2,1]$, which is not a lattice.
\end{eg}

Agliardi \cite[Proposition 4]{agliardi2000generalization} proves that for a game satisfying pseudo-supermodularity, single crossing condition and its variant \cite[Condition (A), p.253]{agliardi2000generalization}, as well as topological assumptions, there is a largest and a least Nash equilibria. Theorem \ref{thm:largest} is a purely order-theoretic analog of Agliardi's result. 

Compared with Theorem \ref{thm:existence}, the existence of largest Nash equilibrium in Theorem \ref{thm:largest} is stronger than the mere existence.  The upper-crossing condition is weaker than the modular-crossing condition. Still, the weakly pseudo-supermodular hypothesis in Theorem \ref{thm:largest} is stronger than the  meet-superextremality hypothesis in Theorem \ref{thm:existence}. \begin{df}[{\cite[p.542]{kukushkin2013increasing}}]
	Let $P$ be a poset.   If for every nonempty chain  $C\subset P$, there is $u\in P$ such that for every  $c\in C$, one has $c\le u$ (resp. $c\ge u$), then $P$ is called \emph{chain-bounded above} (resp. \emph{below}).\end{df}
\begin{thm}\label{thm:largest}
	Assume that  for every $i\in N$, \begin{enumerate}\item $S_i$ is a complete lattice;
	\item\label{it:large-extremal}    for every $s_{-i}\in S_{-i}$, the function $u_i(\cdot,s_{-i}):S_i\to C_i$ is weakly pseudo-supermodular;
		\item  the payoff function $u_i:S\to C_i$ is an upper-crossing  function relative to $(S_i,S_{-i})$;
			\item\label{it:large-up}  for  every $s_{-i}\in S_{-i}$ and every $t\in C_i$, the subset $[u_i(\cdot,s_{-i})\ge t]$ is chain-complete downwards  and chain-bounded above. 
	\end{enumerate}Then the game admits a largest Nash equilibrium.
\end{thm}
\begin{proof}
	From Zorn's lemma, a chain-bounded above poset admits maximal elements.	Then by Assumptions \ref{it:large-extremal}, \ref{it:large-up} and Corollary \ref{cor:41} \ref{it:wsup}, for every $i\in N$ and every $s\in S$, the subset  $R_i(s)$ is  nonempty. By Assumption \ref{it:large-up}, it is also chain-bounded above.  
	
We claim that $R:S\to 2^S$ is upper C-ascending in the sense of \cite[Definition 2.3]{yu2023generalization2}. For this,	consider every $i\in N$, any $x_{-i}\le x'_{-i}$ in $S_{-i}$, every $x_i\in R_i(x_{-i})$ and every $x_i'\in R_i(x'_{-i})$. Since $u_i(x_i\wedge x'_i,x_{-i})\le u_i(x_i,x_{-i})$ and $u_i(\cdot,x_{-i})$ is weakly pseudo-supermodular,  one has \[u_i(x_i\vee x_i',x_{-i})\ge u_i(x_i,x_{-i})\wedge u_i(x'_i,x_{-i})=u_i(x'_i,x_{-i}).\] By upper-crossing property of $u_i$, there exists $x''_i\in S_i$ with $x''_i \ge x_i\vee x'_i$  and $u_i(x'_i,x_{-i}')\le u_i(x_i'',x_{-i}')$. Thus $x_i''\in R_i(x'_{-i})$. For any $x\le x'$ in $S$, every $y\in R(x)$, every $y'\in R(x')$ and every $i\in N$,  there is $z_i\in R_i(x'_{-i})$ with $z_i\ge y_i\vee y'_i$. Then $z=(z_i)_i$ is an element of $R(x')$ with $z\ge y\vee y'$. 
	
	The claim together with \cite[Lemma 3.6]{yu2023generalization2} yields the result.
\end{proof}
In Theorem \ref{thm:largest}, the set of equilibria may not be a lattice nor posses a least element, as Example \ref{eg:hasmaxnomin} illustrates. 
\begin{eg}\label{eg:hasmaxnomin} 
	Let $N=\{1,2\}$, $S_i=[0,1]$. Define $u_1:S\to \R$ by \[u_1(s_1,s_2)=\begin{cases}
		s_1& \text{ if }s_2\le 1/2,\\
		0& \text{ if }s_2>1/2.
	\end{cases}\] Define $u_2:S\to \R,\quad s\mapsto 0$. Then the corresponding game satisfies the conditions of Theorem \ref{thm:largest}. The best reply correspondence is \[R:S\to 2^S,(s_1,s_2)\mapsto \begin{cases}
		\{1\}\times [0,1] & \text{ if }s_2\le 1/2;\\
		[0,1]^2 & \text{ if }s_2>1/2.
	\end{cases}\]  The set of Nash equilibria is $\{1\}\times [0,1/2]\cup[0,1]\times (1/2,1]$. The largest Nash equilibrium is $(1,1)$, but there is no least Nash equilibrium. A minimal Nash equilibrium is $(1,0)$. By \cite[Lem.~2.1]{yu2023topkis2}, the  set of Nash equilibria is not a lattice.
\end{eg}
\subsection*{Completeness}
Theorem \ref{thm:structure} is purely order-theoretic. In particular,  we do not require   topological conditions \cite[(1), (2), p.175]{milgrom1994monotone}.  As Example \ref{eg:01ab} \ref{it:wkqsnotqs} shows, the weak quasisupermodularity hypothesis is strictly weaker than quasisupermodularity assumed in \cite[(3), p.175]{milgrom1994monotone}. The assumption \ref{it:uplevelcpt<usc} is weaker than the condition of \cite[Theorem A4]{milgrom1994monotone}. Moreover, the completeness of the Nash equilibria set is stronger than the existence of largest and smallest Nash equilibria established in \cite[Theorem 12]{milgrom1994monotone}. 
\begin{thm}\label{thm:structure}
	Assume that  for every $i\in N$, \begin{enumerate}\item $S_i$ is a complete lattice;
	\item\label{it:quasisupermodular} for every $s_{-i}\in S_{-i}$, the payoff function $u_i(\cdot,s_{-i}):S_i\to C_i$ is weakly quasisupermodular;
		\item\label{it:singlecrossing} the function $u_i:S\to C_i$ satisfies the single crossing property relative to $(S_i,S_{-i})$;
		\item\label{it:uplevelcpt<usc}  for every $s_{-i}\in S_{-i}$ and every $t\in C_i$, the set	 $[u_i(\cdot,s_{-i})\ge t]$ is chain-complete downwards (resp. upwards)  and  chain-bounded above (resp. below).
	\end{enumerate}
	Then the set of Nash equilibria is a nonempty complete lattice. 
\end{thm}

\begin{proof}By symmetry, it suffices to prove the statement without parentheses.	We claim that for every $i\in N$, the correspondence $R_i:S_{-i}\to 2^{S_i}$ is increasing in the sense of \cite[p.33]{topkis1998supermodularity}. For this, consider any $x_{-i}\le x'_{-i}$, every $x_i\in R_i(x_{-i})$ and every $x'_i\in R_i(x'_{-i})$. 
	\begin{itemize}\item One has $x_i\wedge x'_i\in R_i(x_{-i})$. Otherwise, $u_i(x_i\wedge x'_i,x_{-i})<u_i(x_i,x_{-i})$. As $u_i(\cdot,x_{-i})$ is weakly quasisupermodular, one has $u_i(x_i\vee x'_i,x_{-i})>u_i(x_i,x_{-i})\wedge u_i(x'_i,x_{-i})=u_i(x'_i,x_{-i})$. Since $u_i$ has the single crossing property, one has $u_i(x_i\vee x'_i,x'_{-i})>u_i(x'_i,x'_{-i})$, which contradicts $x'_i\in R_i(x'_{-i})$;
		\item One has $x_i\vee x'_i\in R_i(x'_{-i})$. Otherwise, $u_i(x_i\vee x'_i,x'_{-i})<u_i(x'_i,x'_{-i})$. As $u_i(\cdot,x'_{-i})$ is weakly quasisupermodular, one has $u_i(x_i\wedge x'_i,x'_{-i})>u_i(x_i,x'_{-i})\wedge u_i(x'_i,x'_{-i})=u_i(x_i,x'_{-i})$. Since
		$u_i$ has the single crossing property, one has $u_i(x_i\wedge x'_i,x_{-i})>u_i(x_i,x_{-i})$, which contradicts $x_i\in R_i(x_{-i})$.\end{itemize}
		The claim is proved. Then the correspondence $R:S\to 2^S$ is also increasing. In particular, for every $s\in S$, the value $R(s)$ is a sublattice of $S$.
	
By Corollary \ref{cor:41} \ref{it:wsup},	for every $i\in N$ and every $s\in S$,  the set $R_i(s)$ is  \emph{nonempty}. By assumption, $R_i(s)$ is  chain-complete downwards and chain-bounded above. Therefore, for every $s\in S$, the lattice $R(s)$ is  nonempty, chain-complete  downwards  and chain-bounded above.  By Lemma \ref{lm:MilgromRoberts}, it is a complete lattice. From  \cite[Theorem 1.3]{yu2023generalization1}, $\Fix(R)$ is a nonempty complete lattice.
\end{proof}
Lemma \ref{lm:MilgromRoberts} below is used in the proof of Theorem \ref{thm:structure}. It is stronger than \cite[Footnote 8]{milgrom1990rationalizability}, which strengthens the chain-bounded above hypothesis to chain-complete upwards condition  and is  stated without proof. 
\begin{lm}\label{lm:MilgromRoberts}
	A  lattice is  complete if and only if it is chain-complete downwards and chain-bounded above.
\end{lm}
\begin{proof}
	The ``only if'' part is by definition. To prove the ``if'' part, let $X$ be a chain-complete downwards and chain-bounded above lattice. 
	
	We claim that $X$ is chain-complete upwards. For every nonempty chain $C\subset X$, since $X$ is chain-bounded above, the subset $X':=\{x\in X|c\le x,\forall c\in C\}$ is nonempty. For any $a,b\in X'$ and every $c\in C$, one has $c\le a$ and $c\le b$, so $c\le a\wedge b\le a\vee b$. Thus, $a\wedge b, a\vee b\in X'$. In particular, $X'$ is a sublattice of $X$.  As $X$ is chain-complete downwards, for every nonempty chain $C'\subset X'$, the element $\inf_XC'$ exists. For every $c\in C$ and every $c'\in C'$, one has $c\le c'$, so $c\le \inf_XC'$. Thus, $\inf_XC'\in X'$ is a lower bound on $C'$. By Zorn's lemma, $X'$ has a minimal element $m$. From \cite[Lem.~2.1]{yu2023topkis2}, $m$ is the least element of  $X'$, so $m=\sup_XC$. The claim is proved. 
	
	By Veinott's lemma (see, e.g., \cite[Lemma 2.7]{yu2023generalization1}) and the claim, $X$ is complete.
\end{proof}

Theorem \ref{thm:structure} generalizes Zhou's theorem \cite[Theorem 2]{zhou1994set}.
\begin{cor}\label{cor:Zhoucomplete}
	Let	$(N,\{S_i\},\{C_i\},\{u_i\})$ be a normal form game.
	Assume that for every $i\in N$,  \begin{enumerate}
		\item  $S_i$ is a lattice with a topology $\tau_i$ finer than the interval topology;
		\item the chain $C_i$ is $\R$;
		\item  for every $s_{-i}\in S_{-i}$, the function $u_i(\cdot,s_{-i}):S_i\to \R$ is upper semicontinuous in $\tau_i$;
		\item for every $s_{-i}\in S_{-i}$,  the function $u_i(\cdot,s_{-i}):S_i\to \R$ is supermodular; 
		\item  the payoff function $u_i:S\to \R$ has increasing differences in $S_i$ and $S_{-i}$.
	\end{enumerate}
	Then the set of Nash equilibria is a nonempty complete lattice. 
\end{cor}
\begin{proof}
	By the Frink-Birkhoff theorem  \cite[Thm.~20, p.250]{birkhoff1940lattice}, every $S_i$ is complete lattice. Supermodular functions are weakly quasisupermodular. Functions having increasing differences satisfy the single crossing property.  By upper semicontinuity, for every $i\in N$, every $s_{-i}\in S_{-i}$ and every $t\in \R$, the set	 $[u_i(\cdot,s_{-i})\ge t]$ is closed in $(S_i,\tau_i)$, so compact.  The interval topology of $[u_i(\cdot,s_{-i})\ge t]$ is coarser than  the subspace topology $\tau_i|_{[u_i(\cdot,s_{-i})\ge t]}$, so also compact. By Lemma \ref{lm:posetFrink}, the poset $[u_i(\cdot,s_{-i})\ge t]$ is chain-complete. Then we conclude by Theorem \ref{thm:structure}. 
\end{proof}
\begin{lm}\label{lm:posetFrink}
	Let $P$ be a poset endowed with the interval topology. If $P$ is compact, then it is chain-complete. 
\end{lm}\begin{proof}
	Let $C\subset P$ be a nonempty chain. Let $B:=\{b\in P|b\le c,\forall c\in C\}$. Consider the family of closed subsets $\{(-\infty,c]\}_{c\in C}$ of $P$. Since $C$ is a chain, for every finite subfamily $\{(-\infty,c_i]\}_{i=1}^n$, there is $1\le i_0\le n$ with $c_{i_0}\le c_i$ for all $1\le i\le n$. Then $c_0\in \cap_{i=1}^n(-\infty,c_i]$. As every finite subfamily has nonempty intersection, by compactness,   $B:=\cap_{c\in C}(-\infty,c]$ is nonempty. Consider the family of closed subsets $\{[b,c]\}_{b\in B,c\in C}$. By the same argument,  $\cap_{b\in B,c\in C}[b,c]$ has an element, which is $\inf_PC$. By symmetry, $\sup_PC$ also exists. Therefore, $P$ is chain-complete.
\end{proof}A chain-complete poset may not be compact in the interval topology, as Example \ref{eg:chain-completenotcompact} shows.\begin{eg}\label{eg:chain-completenotcompact}
	Let $P=\{x_i\}_{i\ge1}\sqcup\{y_j\}_{j\ge 1}$. Define a partial order $\le$ such that the only nontrivial relations are as follows. For any positive integers $j\le i$, set $x_i\le y_j$. Then every chain in $P$ has at most two elements, so $P$ is chain-complete. Endow $P$ with the interval topology. The family of closed subsets $\{(-\infty,y_j]\}_{j\ge 1}$ has finite intersection property, since for any finitely many indices $j_1,\dots,j_n$, there is an integer $i$ with $i\ge j_k$ for every $1\le k\le n$. Then $x_i$ is in the corresponding finite intersection. However, the intersection $\cap_{j\ge 1}(-\infty,y_j]$ is empty. Thus, $P$ is not compact.  
\end{eg}Example \ref{eg:cannotuseZhoufixed} shows that Theorem \ref{thm:structure} is a \emph{proper} generalization of Corollary \ref{cor:Zhoucomplete}.
\begin{eg}\label{eg:cannotuseZhoufixed}
	Let $N=\{1,2\}$, $S_i=[0,1]$. Define $u_1:S\to \R$ by \[u_1(s_1,s_2)=\begin{cases}
		1 &\text{ if } s_1\in [0,1/2)\cup\{1\},\\
		0 &\text{ if } s_1\in [1/2,1).
	\end{cases}\] Define $u_2:S\to \R,\quad s\mapsto 0$. Since $R(0,0)=([0,1/2)\cup\{1\})\times [0,1]$ is \emph{not}  subcomplete in $S$, one cannot apply Zhou's fixed point theorem \cite[Theorem 1]{zhou1994set} in this case. Still, the game satisfies the conditions of Theorem \ref{thm:structure}. For example, because $u_1$ is independent of $s_2\in S_2$, Condition \ref{it:singlecrossing} is verified. The set of Nash equilibrium $([0,1/2)\cup\{1\})\times [0,1]$ is a complete lattice. 
\end{eg}
\printbibliography
\end{document}